\documentclass[runningheads]{llncs}

\usepackage[T1]{fontenc}
\usepackage{parskip}
\usepackage{amsmath}
\usepackage{amssymb}
\usepackage{mathtools}
\usepackage{thmtools}
\usepackage{bm}
\usepackage{dsfont}
\usepackage[mathscr]{euscript}
\usepackage{amsfonts}

\usepackage[pdfpagelabels,bookmarks=false]{hyperref}
\usepackage[capitalize, nameinlink]{cleveref}
\usepackage{enumitem}

\usepackage{float}
\usepackage{caption}
\usepackage{subcaption}

\usepackage{xurl} % Needed for correctly breaking bibliography links
\usepackage[backend=bibtex, maxbibnames=10, maxcitenames=10]{biblatex}
\usepackage{tikz}
\usetikzlibrary{positioning}
\usetikzlibrary{calc}
\usetikzlibrary{decorations.markings}
\usetikzlibrary{decorations.pathreplacing}
\usetikzlibrary{arrows.meta}
\usetikzlibrary{shapes.geometric} 
\usepackage[dvipsnames]{xcolor} 
\usepackage{pgfplots}        
\usepgfplotslibrary{fillbetween}

\pgfplotsset{compat=1.18}

\usepackage[linesnumbered,ruled,vlined]{algorithm2e}
\SetKwInput{KwInput}{Input}
\SetKwInput{KwOutput}{Output}
\DontPrintSemicolon

\foreach \x in {A,...,Z}{
	\expandafter\xdef\csname b\x\endcsname{\noexpand\mathbb{\x}}
	\expandafter\xdef\csname c\x\endcsname{\noexpand\mathcal{\x}}
    \expandafter\xdef\csname s\x\endcsname{\noexpand\mathcal{\x}}
	\expandafter\xdef\csname f\x\endcsname{\noexpand\mathfrak{\x}}
}

\newcommand{\stackalign}[2]{\stackrel{\mathclap{#1}}{#2}}
\renewcommand{\bar}[1]{\overline{#1}}
\renewcommand{\hat}[1]{\widehat{#1}}

\newcommand{\ceil}[1]{\lceil #1 \rceil}
\newcommand{\abs}[1]{\left\lvert#1\right\rvert}
\newcommand{\norm}[1]{\left\lVert#1\right\rVert}

\newcommand{\lS}{\boldsymbol{\mathscr{S}}}
\newcommand{\lT}{\boldsymbol{\mathscr{T}}}

\newcommand{\Opt}{\mathrm{OPT}}
\newcommand{\single}{\mathrm{single}}
\newcommand{\double}{\mathrm{double}}
\newcommand{\epsbucket}{\epsilon_b}

\newcommand{\sigmaeta}{\sigma_\eta}
\newcommand{\sigmathird}{\sigma_{1/3}}
\newcommand{\sigmageneric}{\sigma_\delta}
\newcommand{\cgeneric}{c_\delta}

\newcommand{\remainingtours}{\lT^*_\mathrm{left}}
\newcommand{\discardedtours}{\lT^*_\mathrm{dis}}
\newenvironment{proofof}[1]
{\par\noindent\textit{Proof of #1.}\ }
{\hfill\qed\par}

\title{Improved Approximation for Unsplittable CVRP via a Greedy Approach}
\author{Daniel Ebert and Leonard Weismantel}
\institute{Research Institute for Discrete Mathematics, Bonn, Germany \\
    \email{\{ebert,weismantel\}@dm.uni-bonn.de}}
\date{}

\begin{document}

\maketitle

\begin{abstract}
     We devise a polynomial-time $3.159$--approximation algorithm for the metric unsplittable Capacitated Vehicle Routing Problem. We build on the Relative Greedy Algorithm suggested by Traub \cite{traub2025approximation}, which can be considered as a variant of the LP rounding algorithm of \textcite{FMRS25}. Our main ingredient is the Average Greedy Algorithm, a new algorithm that controls both tour costs and the coverage of clients with high demand. This additional control enables a sharper averaging argument for the cost of subsequent greedy choices. Similarly to \textcite{ZX26}, combining the Average Greedy with variants of tour partitioning and a matching algorithm yields the final approximation guarantee.
\end{abstract}

\section{Introduction}\label{sec:introduction}
The Capacitated Vehicle Routing Problem (CVRP), introduced by \textcite{DR59} in 1959, is one of the fundamental problems in combinatorial optimization, with obvious applications in logistics. In the metric unsplittable variant considered in this paper, we are given a set $V$ of clients, a depot $s$, edge costs $c\colon (V\cup \{s\}) \times (V\cup\{s\}) \to \mathbb R_{\geq 0}$ that satisfy the triangle inequality and demands $d: V \to [0,1]$. A feasible solution is a collection of tours, each starting and ending at the depot, such that every client is served by exactly one tour and the total demand served by each tour is at most one. The objective is to minimize the total cost of these tours.

CVRP is NP-hard, generalizing the metric Traveling Salesman Problem and Bin Packing. Building on a TSP tour partitioning approach of \textcite{HR85}, \textcite{AG87} obtained an $(\alpha +2)$-approximation algorithm, where $\alpha$ denotes the approximation ratio for TSP. Taking $\alpha = 1.5$ yields a $3.5$-approximation algorithm. Only several decades afterward,  \textcite{BTV23} improved this by a small positive $\epsilon$. Later, \textcite{FMRS25} obtained a guarantee of $\alpha + \frac{1}{1-\delta} + \ln(2)$ for any $\delta \in (0, 1/2)$ which gives a $3.194$-approximation for small enough $\delta$. Their approach uses a modified tour splitting procedure that reduces the cost for clients with demand at most $\delta$. Clients with demand at least $\delta$ can be partially served by sampling tours from a configuration LP. More recently, \textcite{ZX26} further improved this to below $3.176$ by combining variants of the LP rounding approach with different tour partitioning procedures and a matching algorithm.

In this paper, we improve the approximation ratio to below $\alpha + 1.659$ which gives $3.159$ in the symmetric case. All of our approaches work for symmetric and asymmetric cost functions. Our improvement is based on an idea by Traub \cite{traub2025approximation}, which we sketch here: Instead of sampling from a configuration LP, one can greedily select tours, minimizing $c(T) / \sigma(T)$. Here, $\sigma$ denotes a cost bound for the additional cost of serving the clients in $T$ in the tour splitting procedure. After having selected tours $S_1,\ldots, S_t$, restricting the tours in an optimum solution to the remaining clients and averaging shows
\begin{equation}\label{eq:relative_greedy_cost_density}
    \min_{T \subseteq V\setminus \bigcup S_i} \frac{c(T)}{\sigma(T)}\leq \frac{\Opt}{\sigma(V) - \sum_{i = 1}^t\sigma(S_i)}\;.
\end{equation}
From this, one can recover the approximation ratio of $\alpha + \frac{1}{1-\delta} + \ln(2)$ achieved by \textcite{FMRS25}; see 
\cref{theorem:traubvygen} below.

In order to improve the guarantee, we note that \eqref{eq:relative_greedy_cost_density} is not always tight. Whenever we select a tour $S_t$ containing a client $v$ with large demand, the quality of the tours in the optimum solution does not degrade uniformly: Restricting to the remaining clients, the tour $T^*_v$ of the optimum solution containing $v$ loses a large fraction of its demand and therefore becomes less efficient. More precisely, $\sigma(T^*_v)$ decreases by at least $\sigma(v)$, while the total $\sigma$-value of \emph{all} other tours in the optimum solution decreases by at most $\sigma(S_t \setminus \{v\})$. Therefore, we can obtain a better guarantee by removing $T^*_v$ from the averaging argument above:
\begin{equation}\label{eq:idea-remove-single-tour-from-average}
\min_{T \subseteq V\setminus \bigcup S_i} \frac{c(T)}{\sigma(T)} \leq \frac{\Opt - c(T_v^*)}{\sigma(V) - \sum_{i = 1}^t\sigma(S_i) - \sigma(T^*_v\setminus \{v\})}\;.
\end{equation}
In \cref{sec:cost-densities}, we will apply this idea to obtain a better bound on the cost of the selected tours by excluding certain subsets of the optimum solution from the average.

This improves the approximation ratio only if we can ensure that the tours selected by the algorithm contain sufficiently many clients with large demand. In \cref{sec:greedy}, we introduce a new greedy algorithm which guarantees, up to an arbitrarily small error, the same cost bound as the algorithm by Traub, but also provides lower bounds on the amount of selected clients with large demand. We combine this greedy algorithm with two different versions of the tour partitioning routine, as well as a matching approach (\cref{sec:matching}), similar to \textcite{ZX26}. Taking the best of the three solutions yields the approximation guarantee, which we compute in \cref{sec:compute-apx-ratio}. The technical computations in \cref{sec:compute-apx-ratio} were done using AI assistance.
We state our main result as follows.
\begin{theorem}\label{thm:main_theorem}
    There exists a polynomial-time $(\alpha + 1.659)$--approximation algorithm for CVRP, where $\alpha$ denotes the approximation ratio for (symmetric or asymmetric) TSP.
\end{theorem}

\subsection{Preliminaries}
\label{subsec:notation-lower-bounds}
For a tour $T$, denote its cost by $c(T)$. We also use $T$ to directly refer to the client set of the tour. For $W\subseteq V$, let $T[W]$ be the tour obtained by short-cutting $T$ to $T\cap W$. All of these definitions are extended in the obvious way to any family $\lT$ of tours. For $v \in V$, define
\[R_0(v) \coloneq c(s,v) + c(v,s), \quad R_1(v) \coloneq \min\{d(v), 0.5\}\cdot R_0(v),\]
extended to vertex sets $W \subseteq V$ via $R_{\{0,1\}}(W) \coloneq \sum_{v \in W}R_{\{0,1\}}(v)$.
Thus, $R_0(v)$ is the length of the singleton tour serving $v$, whereas $R_1(v)$ is its demand-weighted radial cost. To simplify notation, we assume $R_0(v)>0$ for any client $v \in V$ by serving any client violating this with a trivial singleton tour. As observed by \textcite{HR85},
\[R_1(V) \leq \Opt\;.\]
For $0 < l\leq u$, write
\[
    V_l^u:=\{v\in V \mid l<d(v)\leq u\}
\]
and $V_0^u=\{v:d(v)\leq u\}$. 

Throughout this paper, fix an optimal solution $\lT^*$.  A client in
$V_{1/3}^1$ is called \emph{single} if its tour in $\lT^*$
contains exactly one client from $V_{1/3}^1$, and \emph{double} if that
tour contains exactly two. Denote by $V_{\single}$ and $V_{\double}$ the
corresponding client sets and by $\lT^*_\single$ and $\lT^*_\double$ the tours in $\lT^*$ containing a single and double client, respectively. Note that $V_{1/3}^1 = V_\single\;\dot\cup\;V_\double$ and
\begin{equation}\label{eq:single-double-radial}
    R_0(V_{\single})+\frac12R_0(V_{\double})\leq\Opt\;.
\end{equation}
Our algorithms are based on the following variant of the iterated tour partitioning algorithm, which in the following is denoted by $\delta$-ITP. Recall that $\alpha$ denotes the approximation ratio for (symmetric or asymmetric) TSP.
\begin{lemma}[\textcite{FMRS25}]
\label{lem:delta_tank}
For $\delta \in (0,\tfrac12)$, define the \emph{potential} of a client $v \in V$ as
\[
    \sigma_\delta(v) \coloneq
    \begin{cases}
        \displaystyle\tfrac{d(v)}{1-\delta}\cdot R_0(v), \quad&v\in V_0^\delta,\\[0.6em]
        \displaystyle\tfrac{2d(v)-\delta}{1-\delta}\cdot R_0(v), \quad&v\in V_\delta^{1/2},\\[0.6em]
        R_0(v), &v \in V_{1/2}^1\;.
    \end{cases}
\]
There is a polynomial-time algorithm $\delta$-ITP that, given any $W \subseteq V$, computes a CVRP solution that serves the clients in $W$ with cost at most $\alpha\cdot\Opt + \sigmageneric(W)$.
\end{lemma}
This is proved in \cref{app:proof-delta-itp}.
As noted earlier, the $\delta$-ITP provides a good solution for clients with small demand. We next give a description of the Relative Greedy Algorithm due to Traub \cite{traub2025approximation} that combines the $\delta$-ITP algorithm with an initial greedy phase, to serve a subset of clients with demand at least $\delta$.

\begin{algorithm}[ht]
	Set $A_0 \gets V_\delta^1$ and $t \gets 0$\\
	\While{There is a tour $ T\subseteq A_t$ such that $c(T) \leq \sigmageneric(T)$}
	{Let $S_{t+1} \subseteq A_t$ be a non-empty tour that minimizes $c(T)/\sigmageneric(T)$ \\
    $A_{t+1}\gets A_t\setminus S_{t+1}$,\quad $t \gets t+1$}
	Apply the $\delta$-ITP Algorithm to the remaining clients $A_t \cup V_0^\delta$ \\
	\textbf{Return} the resulting solution
	\caption{Relative Greedy for CVRP}
    \label{algorithm:relativegreedy}
\end{algorithm}

\begin{theorem}[Traub \cite{traub2025approximation}]\label{theorem:traubvygen}
	Algorithm \ref{algorithm:relativegreedy} is a polynomial-time $(\alpha + \frac{1}{1 - \delta}+ \ln(2))$ -- approximation algorithm for any fixed $\delta \in (0, 1/2)$.
\end{theorem}
\begin{proof}
The runtime is dominated by the enumeration of candidate tours in the greedy phase. Since every such tour contains at most $1 / \delta$ clients, this can be done in polynomial time.

It remains to prove the approximation guarantee. Fix iteration $t$ and let $A_t$ be the set of remaining large clients. An averaging argument gives
\[\min_{T \in \lT^*} \frac{c(T)}{\sigmageneric(T[A_t])} 
\le \frac{\sum_{T \in \lT^*} c(T)}{\sum_{T \in \lT^*} \sigmageneric(T[A_t])}
=
\frac{\Opt}{\sigmageneric(A_t)}.\]
Since the algorithm chooses a tour $S_t$ of minimum cost-to-potential ratio of at most $1$, we have
\[c(S_t)\le \sigmageneric(S_t)\min\left\{1,\frac{\Opt}{\sigmageneric(V)-\sum_{i < t}\sigmageneric(S_i)}
\right\}.\]
Thus, the total cost of serving large clients $V_\eta^1$ in the greedy phase and the following $\delta$-ITP phase are at most 
\[
\int_0^{\sigmageneric(V)}\min\left\{\frac{\Opt}{\sigmageneric(V)-x}, 1\right\}\,dx
\leq (1+\max\{0, \ln(\gamma)\})\Opt,
\]
where $\gamma \coloneq 2R_1(V_\delta^1) / \Opt \geq \sigmageneric(V) / \Opt$.
By \cref{lem:delta_tank}, the cost for serving the clients $V_0^\delta$ is at most
$R_1(V_0^\delta) / (1-\delta)$, where
\[R_1(V_0^\delta) = R_1(V) - (\gamma / 2)\cdot \Opt \leq (1-\gamma/2)\cdot\Opt\;.\]
Including the cost of the TSP tour, we obtain an approximation ratio of
\[\alpha + 1 + \max\left\{0, \ln(\gamma)\right\} + \frac{1 - \gamma/2}{1-\delta}\;.\]
This expression is maximized for $\gamma = 2 - 2\delta$, which concludes the proof. \qed
\end{proof}

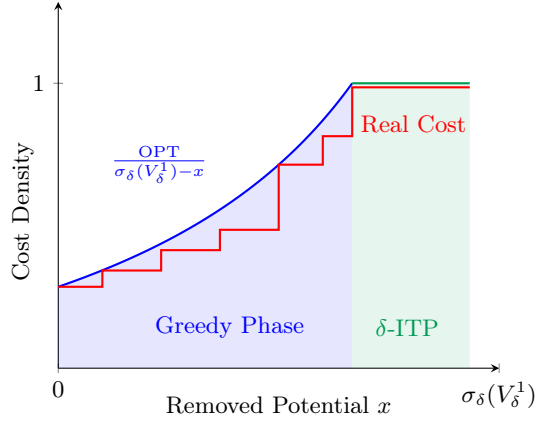
\begin{figure}[t]
    \centering
     \begin{tikzpicture}[scale = 0.85]
            \begin{axis}[
                axis lines = left,
                xlabel = {Removed Potential $x$},
                xlabel shift = -10pt,
                ylabel = {Cost Density},
                ylabel shift = -6pt,
                ylabel style={xshift=-10pt},
                ymin=0.3, ymax=1.2, xmin=0, xmax=1.5,
                xtick={0,1.5}, xticklabels={0, $\sigmageneric(V_\delta^1)$},
                ytick={1}, yticklabels={1},
                title={}
            ]
                % Worst case curve
                \addplot[domain=0:1, samples=100, thick, blue] {1/(2-x)};
                \node[blue, right, font=\footnotesize] at (axis cs: 0.15, 0.8) {$\frac{\Opt}{\sigmageneric(V_\delta^1)-x}$};
                \node[blue, right, font=\footnotesize] at (axis cs: 0.3, 0.4) {Greedy Phase};
                \addplot[domain=1:1.4, samples=100, thick, ForestGreen] {1};
                % \node[ForestGreen, right, font=\footnotesize] at (axis cs: 1, 1.1) {$1$};
                \node[ForestGreen, right, font=\footnotesize, align=left] at (axis cs: 1.05, 0.4) {$\delta$-ITP};

                % Real cost staircase function
                \addplot[const plot, thick, red] coordinates {
                    (0,0.5) (0.15,0.54) (0.35,0.59) (0.55,0.64) (0.75,0.8) (0.9,0.87) (1.0,0.99) (1.2,0.99) (1.4,0.99)
                };
                %(0,0.5) (0.15,0.54) (0.35,0.605) (0.55,0.689) (0.75,0.8) (0.9,0.909) (1.0,0.99) (1.2,0.99) (1.4,0.99)
                \node[red, right, font=\footnotesize] at (axis cs: 1, 0.9) {Real Cost};

                % Area
                \addplot[name path=f, domain=0:1, samples=100, opacity=0] {1/(2-x)};
                
                \path[name path=axis] (axis cs:0,0.3) -- (axis cs:1,0.3);
                \addplot [
                    thick,
                    color=blue,
                    fill=blue, 
                    fill opacity=0.1
                ]
                fill between[
                    of=f and axis,
                    soft clip={domain=0:1},
                ];
                \addplot[name path=g, domain=1:1.4, samples=100, opacity=0] {1};
                \path[name path=axis] (axis cs:1,0.3) -- (axis cs:1.4,0.3);
                \addplot [
                    thick,
                    color=ForestGreen,
                    fill=ForestGreen, 
                    fill opacity=0.1
                ]
                fill between[
                    of=g and axis,
                    soft clip={domain=1:1.4},
                ];

            \end{axis}
        \end{tikzpicture}
\caption{\small Illustration of the cost-to-potential ratio over the course of Algorithm \ref{algorithm:relativegreedy}. The greedy phase ends when this ratio is equal to 1.}\label{fig:cost_density}
\end{figure}

\subsection{The Matching Savings Algorithm}\label{sec:matching}
In this section we provide an algorithm similar to the matching algorithm by \textcite{FMRS25}. This algorithm will provide the guarantee that either itself is already good enough or the instance contains a lot of single clients.

Let $H$ be the graph with vertex set $V_{1/3}^1$ and an edge $\{v,w\}$ whenever $d(v) + d(w) \leq 1$. For each edge $\{v,w\}$, denote the cost of its two client tour by $c_{\{v,w\}}$. Define its savings as
\[
    u(v,w):=\sigma_{1/3}(v)+\sigma_{1/3}(w)-c_{\{v,w\}}.
\]
The matching savings algorithm computes a maximum weight matching in $(H,u)$ to maximize the savings. 
\begin{algorithm}\caption{Matching Savings for CVRP}\label{alg:matching}
    Initialize $\lT \gets \emptyset$.\\
    Compute a maximum weight matching $M$ in $(H, u)$.\\
    \For{$(v,w)\in E(M)$}{Add the tour serving $v$ and $w$ to $\lT$}
    Apply the $\frac13$-ITP Algorithm to $V\setminus V(M)$ and add the resulting tours to $\lT$.\\
    \textbf{Return} $\lT$
\end{algorithm}
\begin{lemma}\label{lem:MatchingSavingsCostBound}
    Algorithm~\ref{alg:matching} produces a feasible CVRP solution of cost at most
    \[\left(\alpha +\tfrac{3}{2}\right) \Opt + \tfrac{3}{2}(R_1(V_\single)-R_1(V_\double)) - \tfrac{1}{2}(R_0(V_\single) - R_0(V_\double)) \]
\end{lemma}
\begin{proof}
    By \cref{lem:delta_tank}, the cost of the produced solution is at most
    \[\alpha \cdot \Opt + \sigmathird(V \setminus V(M)) + \sum_{\{v,w\} \in E(M)}c_{\{v,w\}}.\]
    Since Algorithm~\ref{alg:matching} maximizes the total saving, this upper bound is no larger than the corresponding expression for any other matching in $H$.
    For $T \in \lT^*_\double$ denote by $v_T, w_T$ its double clients. Note that they form a matching. Therefore, the cost of the produced solution is at most
     \begin{equation}\label{eq:cost-double-matching}
         \alpha \cdot \Opt + \sigmathird(V\setminus V_\double) + c(\lT^*_\double).
     \end{equation}
     Using $R_0(v) \leq c(T)$ for $v \in T$, we compute for $T \in \lT^*_\double$:
     \begin{align*}
         \sigmathird(T \setminus V_{1/3}^1) &\leq \tfrac{3}{2}(1 - d(T \cap V_{1/3}^1))c(T)\\[0.4em]
         &\leq \tfrac{1}{2}c(T) -\tfrac{3}{2}R_1(T \cap V_{1/3}^1) + \tfrac{1}{2}R_0(T \cap V_{1/3}^1)
     \end{align*}
    Inserting this into \eqref{eq:cost-double-matching} and applying the definition of $\sigmathird$ yields a bound of
    \begin{align*}
        \alpha \cdot \Opt &+ \tfrac32 \left(R_1(V\setminus \lT^*_\double) + c(\lT^*_\double)\right)\\[0.4em]
        &+ \tfrac{3}{2}(R_1(V_\single)-R_1(V_\double)) - \tfrac{1}{2}(R_0(V_\single) - R_0(V_\double))
    \end{align*}
    Applying the bound $R_1(T) \leq c(T)$ to every $ T \in \lT^*\setminus \lT^*_\double$ then gives the final result. \qed
\end{proof}

\section{Average Greedy Algorithm}\label{sec:greedy}
In this section, we describe the Average Greedy Algorithm, an improvement of Algorithm \ref{algorithm:relativegreedy} that gives additional guarantees on the amount of selected large clients. Throughout the section, we fix constants $\delta \in (0,\tfrac12)$ and $0 < \epsilon,\eta < \tfrac13$.
Recall that the Relative Greedy Algorithm (described in Algorithm \ref{algorithm:relativegreedy} for the case $\delta=\eta$) selects in iteration $t+1$ a tour $S_{t+1}$ with cost-to-potential ratio
\begin{equation}\label{eq:idea-greedy-cost-density}
    \frac{c(S_{t+1})}{\sigmageneric(S_{t+1})} \leq \frac{\Opt}{\sigmageneric(V_\eta^1) - \sigmageneric(\lS_t)}\;,
\end{equation}
where  $\lS_t$ denotes the set of tours selected in the first $t$ iterations. As sketched in \eqref{eq:idea-remove-single-tour-from-average} for one tour, we may discard any subset of optimum tours $\discardedtours$ from the averaging argument and get
\begin{equation}\label{eq:idea-improved-cost-density}
    \frac{c(S_{t+1})}{\sigmageneric(S_{t+1})} \leq \frac{\Opt - \sum_{T \in \discardedtours} c(T)}{\sigmageneric(V_\eta^1) - \sigmageneric(\lS_t) - \sum_{T \in \discardedtours} \sigmageneric(T \setminus \lS_t)}\;.
\end{equation}
For an optimum tour $T^*$, this is particularly useful if a client $v$ with large demand has already been selected from it. 
In that case, assuming that all clients in the tour have approximately the same distance to the depot, we subtract $c(T^*) \geq R_0(v)$ from the numerator and $\sigmageneric(T^*\setminus \lS_t) \leq \sigmageneric(T^* \setminus \{v\}) \lessapprox 2(1- d(v))\cdot R_0(v)$ from the denominator. Thus, if $1 / 2(1-d(v))$ exceeds the bound in \eqref{eq:idea-greedy-cost-density}, discarding $T^*$ indeed yields an improvement.

By this argument, \eqref{eq:idea-improved-cost-density} yields a significantly better bound than \eqref{eq:idea-greedy-cost-density} if we can guarantee that the previously selected tours contain a sufficient amount of clients with large demand. More precisely, we will guarantee the existence of a subset $B \subseteq \lS_t \cap V_{1/3}^1$ with
\begin{equation}\label{eq:R_expectations}
    R_0(B) \approx \frac{\sigmageneric(\lS_t)}{\sigmageneric(V_\eta^1)}\cdot R_0(V_\single)\;, \qquad
    R_1(B) \approx \frac{\sigmageneric(\lS_t)}{\sigmageneric(V_\eta^1)}\cdot R_1(V_\single)\;.
\end{equation}
We defer to \cref{lemma:r0-r1-close-to-avg} for the exact statement\footnote{Note that the guarantees in \eqref{eq:R_expectations} can only be given for single clients: For a previously selected client $v \in \lS_t \cap V_{1/3}^1$, we (possibly) remove the optimum tour $T^*_v$ containing it from the average to obtain an improved cost bound. Therefore, if $v \in V_\double$ and $T^*_v \cap V_{1/3}^1 = \{v,w\}$, we can not guarantee to serve $w$ while maintaining the improved cost bound \eqref{eq:idea-improved-cost-density}. Observe that for instances with many double clients, Algorithm \ref{alg:matching} produces a very good solution.}. 
Combining this with \eqref{eq:idea-improved-cost-density} and the observations from the previous paragraph, this yields a bound of
\begin{equation}\label{eq:idea-cost-density-single-clients}
    \frac{c(S_{t+1})}{\sigmageneric(S_{t+1})} \lessapprox \frac{\Opt - \frac{\sigmageneric(\lS_t)}{\sigmageneric(V_\eta^1)}R_0(V_\single)}{\sigmageneric(V_\eta^1) - \sigmageneric(\lS_t) - \frac{\sigmageneric(\lS_t)}{\sigmageneric(V_\eta^1)}(2R_0(V_\single) - 2R_1(V_\single))}\;,
\end{equation}
assuming for now that all clients in $B$ come from different tours in the optimum solution. We will show in \cref{sec:cost-densities} that the simplifying assumptions made here are not necessary, and further improve the bound for $\delta = 1/3$.
The bound \eqref{eq:idea-cost-density-single-clients} is useful because if the matching algorithm does not perform well, we must have a significant amount of single clients with large demand by \cref{lem:MatchingSavingsCostBound}. At least in the beginning of the algorithm, when \eqref{eq:idea-greedy-cost-density} is in the worst case close to $1/2$, \eqref{eq:idea-cost-density-single-clients} provides a significantly better bound. Inserting this into the proof of \cref{theorem:traubvygen} would immediately improve the approximation ratio.

In the following, we describe a greedy algorithm that maintains \eqref{eq:R_expectations} while simultaneously guaranteeing that the cost-to-potential ratio of the selected tours is close to optimum. In particular, this will imply approximate versions of \eqref{eq:idea-greedy-cost-density} and \eqref{eq:idea-improved-cost-density}, and therefore \eqref{eq:idea-cost-density-single-clients}. To ensure \eqref{eq:R_expectations}, we divide the single clients into buckets based on their demand:
\begin{definition}[Bucket mass]
    Let $\epsbucket > 0$. For $W\subseteq V$ and $b=1,\ldots,\ceil{1/\epsbucket}$, define
    \[
        r_b(W)
        \coloneq
        R_0\left(
            W\cap
            V_{1/3+(b-1)\cdot(2/3)\epsbucket}
             ^{1/3+b\cdot(2/3)\epsbucket}
        \right).
    \]
\end{definition}
With foresight, we set $\epsbucket \coloneq \epsilon^{1/3}$. For each bucket $b$, we guarantee
\begin{equation}\label{eq:greedy-rb-bound-informal}
r_b(\lS_t) \gtrapprox \frac{\sigmageneric(\lS_t)}{\sigmageneric(V_\eta^1)} \cdot r_b(V_\single)\;.
\end{equation}
By taking a subset of the clients in each bucket if the inequality is too strict\footnote{
    A strict inequality can be a problem because the bound \eqref{eq:idea-improved-cost-density} is not always better than \eqref{eq:idea-greedy-cost-density}, particularly when removing tours that correspond to selected clients with demand close to $1/3$. Taking a subset which satisfies the bound with approximate equality therefore simplifies the analysis.
}, one can then obtain a set $B \subseteq \lS_t \cap V_{1/3}^1$ satisfying \eqref{eq:R_expectations}.

For maintaining this guarantee, which depends on the values $r_b(V_\single)$, the algorithm needs to know them approximately. Since $\epsilon$ is a constant, we can enumerate a constant number of guesses $\bigl(\tilde r_b\bigr)_{b=1}^{\lceil 1/\epsbucket\rceil} \geq 0$ such that one of these guesses satisfies
\begin{equation}\label{eq:rb-guess-accuracy}
    r_b(V_\single) - \epsilon\cdot \sigmageneric(V_\eta^1) \leq \tilde r_b \leq r_b(V_\single) \qquad \text{for every } b=1,\ldots,\left\lceil 1/{\epsbucket}\right\rceil.
\end{equation}
By running the algorithm for every guess vector and returning the cheapest solution of the successful runs, we may assume throughout the analysis that such a guess is fixed. We denote it by $\tilde r_b(V_\single)$. Additionally, we assume
\begin{equation}\label{eq:any-tour-small-r1}
        R_1(T) \leq \epsilon\cdot R_1(V_\eta^1)
\end{equation}
for any tour $T$ on $V_\eta^1$. This assumption can always be satisfied by creating copies of the instance and is proven in \cref{app:proof-any-tour-small-r1}.

To simplify the analysis, we define for every tour $T$ its truncated cost as
\[\cgeneric(T)\coloneq\min\{c(T),\sigmageneric(T)\}\;.\]
At the end of the greedy phase, we will simply dissolve all selected tours $T$ with $c(T) > \sigmageneric(T)$ and serve their clients in the $\delta$-ITP routine, which results in an additional cost of at most $\sigmageneric(T)$. In the algorithm, we use the following definition:
\begin{definition}[Deviation]
    For a cost density $\phi\geq 0$ and a tour $T$, let
    \[u^\phi(T)\coloneq \left(u_c^\phi(T),u_1^\phi(T),\ldots,u_{\ceil{1/\epsbucket}}^\phi(T)\right)\;,\]
    where
    \begin{align*}
        u_c^\phi(T) &\coloneq \cgeneric(T)-\sigmageneric(T)\cdot \phi,\\
        u_b^\phi(T) &\coloneq r_b(T) - \frac{\sigmageneric(T)}{\sigmageneric(V_\eta^1)}\cdot\tilde{r}_b(V_\single), \qquad b=1,\ldots,\ceil{1/\epsbucket}\;.
    \end{align*}
\end{definition}
When setting $\phi$ to the desired bound for the cost-to-potential ratio, $u_c^\phi(T)$ describes the deviation from this target. Similarly, $u_b^\phi(T)$ corresponds to the deviation from the bound in \eqref{eq:greedy-rb-bound-informal}. Of course, it is in general impossible for a single tour to have a small deviation, for example because it cannot contain $\ceil{1/\epsbucket}$ clients with demand at least $1/3$. For this reason, we mostly consider the sum of the deviation vectors:
\begin{definition}
    Fixing selected tours $S_1, \ldots, S_t$ and cost densities $\phi_1, \ldots, \phi_t$, the \emph{discrepancy} is defined by
    \[D(t)\coloneq\sum_{i=1}^t u^{\phi_i}(S_i), \qquad D(0)\coloneq 0\;.\]
    The \emph{one-sided discrepancy} $\bar D(t)$ is defined coordinate-wise by
    \[\bar D_c(t)\coloneq\max\{0,D_c(t)\},\qquad \bar D_b(t)\coloneq\min\{0,D_b(t)\}\]
    for $b=1,\ldots,\ceil{1/\epsbucket}$.
\end{definition}
The one-sided discrepancy retains only unfavorable deviations: If $D_c(t) < 0$, then the selected tours cost less than expected. Similarly, if $D_b(t) > 0$, we have selected more clients in that demand bucket than required, which is also fine. Since the target value $r_b(V_\single)$ does not account for double clients, it might not even be possible to keep $D_b(t)$ small.

In \cref{lemma:greedy-rule}, we will show that $\phi_t$ is not larger than any bound of the form \eqref{eq:idea-improved-cost-density} at each stage of the algorithm. From the definition of $u^\phi$ and $\bar D$, we then obtain approximate versions of \eqref{eq:idea-improved-cost-density} and \eqref{eq:greedy-rb-bound-informal} if $\bar D(t)$ is small. To bound $\bar D(t)$, observe
\[\norm{\bar D(t+1)}_2^2 \leq \norm{\bar D(t)}_2^2 + 2\langle \bar D(t), u^{\phi_{t+1}}(S_{t+1}) \rangle + \norm{u^{\phi_{t+1}}(S_{t+1})}_2^2\;.\]
The algorithm will ensure that the scalar product is always non-positive, which immediately gives a bound on $\bar D(t)$ by induction (see \cref{lemma:discrepancy_bound}). In fact, we simply choose the cost bound $\phi_{t+1}$ as small as possible such that this is possible. A complete description of our algorithm is given in Algorithm \ref{alg:average_relative_greedy}.

\begin{algorithm}[ht]
    \KwInput{CVRP instance $(V,s,c,d)$, parameters $\delta \in (0,\frac{1}{2})$ and $\eta,\epsbucket > 0$, guesses $\tilde r_b(V_\single)$ for $b \in \{1,\ldots,\ceil{1/\epsbucket}\}$}
	Set $A_0 \gets V_\eta^1$ and $t\gets 0$\\
	\While{$A_t \neq \emptyset$}
	{      
        Let $\phi_{t+1} \geq 0$ be minimum such that
        \[\min_{\substack{T \subseteq A_t \\ \text{non-empty tour}}}\langle \bar D(t),u^{\phi_{t+1}}(T) \rangle \leq 0\]
        and let $S_{t+1}$ be a tour attaining this minimum \label{algline:average_greedy_tour_selection}\\
        Set $A_{t+1} \gets A_t \setminus S_{t+1}$ and $t \gets t + 1$
    }
    Dissolve all tours $S_i$ with $c(S_i) > \sigmageneric(S_i)$ and apply the $\delta$-ITP algorithm to serve these clients and those in $V_0^\eta$ \\
	\textbf{Return} the resulting solution
	\caption{Average Greedy for CVRP}\label{alg:average_relative_greedy}
\end{algorithm}
For the remainder of this section, we fix a single instantiation of Algorithm \ref{alg:average_relative_greedy}, and assume that the guesses for $\tilde r_b(V_\single)$ are correct. We denote by $k$ the number of iterations, so the solution before the $\delta$-ITP is given by $S_1, \ldots, S_k$. The set of available clients $A_t$ and the cost density $\phi_t$ for an iteration $t \in \{1,\ldots,k\}$ refer to the values defined in Algorithm \ref{alg:average_relative_greedy}.

\begin{lemma}[Greedy Rule]\label{lemma:greedy-rule}
    For $t \in \{0,\ldots,k-1\}$ and any set $\remainingtours \subseteq \lT^*$ with $\remainingtours \supseteq \left\{ T\in\lT^*: T\cap V_{1/3}^1 \subseteq A_t \right\}$ and $\sigmageneric\bigl(\remainingtours[A_t]\bigr)>0$, we have
    \begin{equation}
        \phi_{t+1} \leq \frac{\cgeneric\bigl(\remainingtours[A_t]\bigr)}{\sigmageneric\bigl(\remainingtours[A_t]\bigr)}\;.
    \end{equation}
\end{lemma}

\begin{proof}
    Let $\sF\coloneq \remainingtours[A_t]$ and $\phi^* \coloneq \cgeneric(\sF) / \sigmageneric(\sF)$. Since the algorithm chooses $\phi_{t+1}$ to be minimum, it suffices to show that there exists a tour $T \in \sF$ such that $\left\langle\bar D(t),u^{\phi^*}(T)\right\rangle\leq 0$.
    For this, we show that $\sum_{T \in \sF}u^{\phi^*}(T)$ has the appropriate sign in every relevant coordinate.

    First, fix a bucket $b$. Every unselected single client belongs to a
    tour of $\remainingtours$: Its optimum tour contains no other
    $1/3$-big client and hence no selected $1/3$-big client. Therefore,
    \[\sum_{T\in\sF}r_b(T) \geq r_b(V_\single)-\sum_{i=1}^t r_b(S_i) \geq \tilde{r}_b(V_\single)-\sum_{i=1}^t r_b(S_i)\;.\]
    Moreover, since the tours in $\sF$ only contain clients that are still available,
    \[\sum_{T\in\sF}\sigmageneric(T) \leq \sigmageneric(V_\eta^1) - \sum_{i=1}^t\sigmageneric(S_i)\;.\]
    It follows that
    \begin{align*}
        \sum_{T\in\sF}u_b^{\phi^*}(T)
        &\geq \tilde{r}_b(V_\single)-\sum_{i=1}^t r_b(S_i) - \left( \sigmageneric(V_\eta^1) - \sum_{i=1}^t\sigmageneric(S_i) \right) \frac{\tilde{r}_b(V_\single)}{\sigmageneric(V_\eta^1)} \\
        &= -\sum_{i=1}^t \left( r_b(S_i) - \sigmageneric(S_i) \frac{\tilde{r}_b(V_\single)}{\sigmageneric(V_\eta^1)} \right) \\
        &=-D_b(t).
    \end{align*}
    If $\bar D_b(t)<0$, then $\bar D_b(t)=D_b(t)$, and hence
    \[\sum_{T\in\sF}u_b^{\phi^*}(T) \geq -\bar D_b(t)>0\;.\]
    If $\bar D_b(t)=0$, the corresponding coordinate does not contribute
    to the scalar product. Thus, in either case,
    \[\bar D_b(t)\sum_{T\in\sF}u_b^{\phi^*}(T)\leq 0\;.\]

    For the cost coordinate, the definition of $\phi^*$ gives
    \[\sum_{T\in\sF}u_c^{\phi^*}(T) = \cgeneric(\sF)-\phi^*\sigmageneric(\sF)\leq 0\;.\]
    Since $\bar D_c(t)\geq 0$, the contribution of the cost coordinate is
    therefore also non-positive. Summing over all coordinates yields
    \[\sum_{T\in\sF} \left\langle\bar D(t),u^{\phi^*}(T)\right\rangle = \left\langle \bar D(t), \sum_{T\in\sF}u^{\phi^*}(T) \right\rangle \leq 0\;.\]
    Since $\sF$ is nonempty, at least one tour $T\in\sF$ satisfies $\left\langle\bar D(t),u^{\phi^*}(T)\right\rangle\leq 0$, which finishes the proof. \qed
\end{proof}
We now use the condition $\langle \bar D(t), u^{\phi_{t+1}}(S_{t+1}) \rangle \leq 0$ to bound the discrepancy:
\begin{lemma}\label{lemma:discrepancy_bound}
    For $t \in \{1,\ldots,k\}$, $\norm{\bar D(t)}_2^2 \leq \sum_{j=1}^t \norm{u^{\phi_j}(S_j)}_2^2$.
\end{lemma}
\begin{proof}
    We proceed by induction on $t$. For $t=0$, this is clear.
    For the induction step, let $I=\{c, r_1,\ldots, r_{\ceil{1/\epsbucket}}\}$ be the set of coordinates of the discrepancy vector.
    Then
    \begin{align*}
        \norm{\bar D(t+1)}_2^2
        &\leq\sum_{\substack{i \in I\text{ s.t.} \\ \bar D_i(t) \neq 0}}(\bar D_i(t)+u^{\phi_{t+1}}_i(S_{t+1}))^2
            + \sum_{\substack{i \in I\text{ s.t.} \\ \bar D_i(t)=0}}(u^{\phi_{t+1}}_i(S_{t+1}))^2 \\
        &=\norm{\bar D(t)}_2^2 + 2\cdot\underbrace{\langle \bar D(t),\ u^{\phi_{t+1}}(S_{t+1}) \rangle}_{\leq 0} + \norm{u^{\phi_{t+1}}(S_{t+1})}_2^2 \\
        &\leq \sum_{j=1}^{t+1} \norm{u^{\phi_j}(S_j)}_2^2.
    \end{align*}
    The last inequality follows from the induction hypothesis and line \ref{algline:average_greedy_tour_selection} in Algorithm \ref{alg:average_relative_greedy}. \qed
\end{proof}

\begin{lemma}\label{lemma:discrepancy_bound_final}
    For $t \in \{1,\ldots,k\}$, $\norm{\bar D(t)}_1 \leq O(\epsilon^{1/3}) \cdot \sigmageneric(V_\eta^1)$.
\end{lemma}
\begin{proof}
    We start by bounding $\norm{u^{\phi_i}(S_i)}_2$ for $i \in \{1,\ldots,t\}$ so that we can apply \cref{lemma:discrepancy_bound}. For this, compute
    \begin{align*}
        \sum_{b=1}^{\ceil{1/\epsbucket}} \left|u^{\phi_i}_b(S_i)\right| &\leq R_0(S_i \cap V_{1/3}^1) + \sigmageneric(S_i) \frac{R_0(V_\single)}{\sigmageneric(V_\eta^1)} \leq 4 \cdot \sigmageneric(S_i)\;,
    \end{align*}
    where we use $\sum_b \tilde r_b(V_\single) \leq R_0(V_\single)$ in the first inequality and $R_0(W) \leq 3 R_1(W) \leq 2 \sigmageneric(W)$ for $W \subseteq V_{1/3}^1$ in the second inequality.
    Furthermore, by Lemma \ref{lemma:greedy-rule}, we have $\phi_t \leq 1$ (note that $c_\delta(T) \leq \sigmageneric(T)$ for any tour $T$ by construction) and therefore
    \[|u^{\phi_t}_c(S_t)| = \abs{\cgeneric(S_t) - \phi_t \cdot \sigmageneric(S_t)} \leq \sigmageneric(S_t)\;.\]
    Combining the above yields
    \begin{align*}
        \norm{u^{\phi_t}(S_t)}_2 &\leq \norm{u^{\phi_t}(S_t)}_1 \leq \sigmageneric(S_t) + 4 \cdot \sigmageneric(S_t) = 5 \cdot \sigmageneric(S_t)\;.
    \end{align*}
    Inserting this into \cref{lemma:discrepancy_bound} yields
    \[\norm{\bar D(t)}_2^2 \leq \sum_{i=1}^t 25\cdot \sigmageneric(S_i)^2 \leq 25\cdot\sigmageneric(\lS_t) \cdot \max_{i \in \{1,\ldots,t\}}\sigmageneric(S_i) \leq 50\cdot\epsilon\cdot \sigmageneric(V_\eta^1)^2\;,\]
    where the last inequality uses the assumption \eqref{eq:any-tour-small-r1}, which implies $\sigmageneric(T) \leq 2\epsilon\sigmageneric(V_\eta^1)$ for any tour $T$. This shows $\norm{\bar D(t)}_2 \leq O(\sqrt{\epsilon}) \cdot \sigmageneric(V_\eta^1)$. Together with the fact that the discrepancy vector has $O(1/\epsbucket)$ entries and $\sqrt{\epsilon}\cdot\sqrt{1 / \epsbucket} = \epsilon^{1/3}$ by definition of $\epsbucket$, we obtain $\norm{\bar D(t)}_1 \leq O(\epsilon^{1/3})\cdot \sigmageneric(V_\eta^1)$. \qed
\end{proof}

By using this bound on the one-sided discrepancy $\bar D$, we can now prove a variant of the informal bound \eqref{eq:R_expectations} we stated earlier:

\begin{lemma}\label{lemma:r0-r1-close-to-avg}
    Let $t \in \{1,\ldots,k\}$. Then there exists a set $B \subseteq \lS_t \cap V_{1/3}^1$ such that
    \begin{align*}
        \abs{R_0(B) - \frac{\sigmageneric(\lS_t) }{\sigmageneric(V_\eta^1)}\cdot R_0(V_\single)} &\leq O(\epsilon^{1/3}) \cdot \sigmageneric(V_\eta^1)\;, \\
        \abs{R_1(B) -   \frac{\sigmageneric(\lS_t)}{\sigmageneric(V_\eta^1)}\cdot R_1(V_\single)} &\leq O(\epsilon^{1/3}) \cdot \sigmageneric(V_\eta^1)\;.
    \end{align*}
\end{lemma}
\begin{proof}
    We choose $B \subseteq \lS_t \cap V_{1/3}^1$ maximal such that we don't exceed the average bucket mass, i.e.
    \[r_b(B) \leq \frac{\sigmageneric(\lS_t)}{\sigmageneric(V_\eta^1)}\cdot \tilde r_b(V_\single)\]
    for all $b \in \{1,\ldots, \ceil{1/\epsbucket}\}$. Since $r_b(v) \leq O(\epsilon) \cdot \sigmageneric(V_\eta^1)$ for $v \in V_{1/3}^1$ by \eqref{eq:any-tour-small-r1}, if this decreased $r_b(B)$, the inequality has slack at most $O(\epsilon) \cdot \sigmageneric(V_\eta^1)$ afterward.
    Combining this with the lower bound on $r_b(\lS_t)$ provided by $\bar D$, and using $\abs{r_b(V_\single) - \tilde r_b(V_\single)} \leq \epsilon\cdot \sigmageneric(V_\eta^1)$ from \eqref{eq:rb-guess-accuracy} yields
    \begin{equation}\label{eq:rb-diff-bound-explicit}
    \begin{aligned}
        \sum_{b=1}^{\ceil{1/\epsbucket}}\abs{r_b(B) - \frac{\sigmageneric(\lS_t)}{\sigmageneric(V_\eta^1)}\cdot  r_b(V_\single)}
        &\leq \norm{\bar D(\lS_t)}_1 + O(\ceil{1/\epsbucket}\cdot \epsilon) \cdot \sigmageneric(V_\eta^1)\\
        &\leq O(\epsilon^{1/3}) \cdot \sigmageneric(V_\eta^1) \;,
    \end{aligned}
    \end{equation}
    where the second bound uses \cref{lemma:discrepancy_bound_final}. Inserting the definition of $r_b$ immediately yields the bound on $R_0$ from the lemma statement. It remains to show the second bound. For any $W \subseteq V_{1/3}^1$, we have
    \begin{align*}
        \abs{\sum_{b=1}^{\ceil{1/\epsbucket}}\alpha_b \cdot r_b(W) - R_1(W)} \leq \epsbucket \cdot R_0(W) \leq 2\epsbucket\cdot \sigmageneric(W),
    \end{align*}
    where $\alpha_b \coloneq \min\{\tfrac13 + (b-1) \cdot \tfrac23 \epsbucket, \tfrac12\} \in [0,\tfrac12]$ is taken from the definition of $r_b$. Combining this with \eqref{eq:rb-diff-bound-explicit}, we obtain
    \begin{equation*}
    \begin{aligned}
    \abs{R_1(B) - \frac{\sigmageneric(\lS_t)}{\sigmageneric(V_\eta^1)}\cdot  R_1(V_\single)} &\leq O(\epsilon^{1/3})\cdot \sigmageneric(V_\eta^1) + 4\epsbucket\cdot \sigmageneric(V_\eta^1) \\
    &\leq O(\epsilon^{1/3}) \cdot \sigmageneric(V_\eta^1)\;,
    \end{aligned}
    \end{equation*}
    which finishes the proof. \qed
\end{proof}

We now proceed to bound the cost of the solution produced by Algorithm \ref{alg:average_relative_greedy}. For this, we need a cost density function which provides an upper bound on the cost-to-potential ratio $\cgeneric(S_t) / \sigmageneric(S_t)$ of the tours selected by the algorithm. By the definition of $\bar D$ and the fact that $\norm{\bar D}_1$ always stays small (\cref{lemma:discrepancy_bound_final}), this essentially corresponds to bounding $\phi_t$. The cost density may depend on the quantities $R_0(B_t),\ R_1(B_t)$, where $B_t \subseteq \lS_t \cap V_{1/3}^1$, and on $\sigmageneric(\lS_t)$. Note that by \cref{lemma:r0-r1-close-to-avg}, the former two quantities can be determined approximately from the third one. For technical reasons, we also require the cost density function to satisfy a certain continuity and monotonicity condition. All this is summarized in the following definition:

\begin{definition}\label{def:cost-density-function}
    For $s \in [0, \sigmageneric(V_\eta^1)]$, let
    \[\bar x(s) \coloneq \left(\frac{s}{\sigmageneric(V_\eta^1)}\cdot R_0(V_\single),\ \frac{s}{\sigmageneric(V_\eta^1)}\cdot R_1(V_\single),\ s\right) \in \bR^3_{\geq0}\;.\]
    A valid \emph{cost density function} is a function
    \[\Phi: [0, R_0(V_{1/3}^1)] \times [0, R_1(V_{1/3}^1)] \times [0, \sigmageneric(V_\eta^1)] \to [0,1]\]
    satisfying the conditions
    \begin{align}
        \phi_{t+1} &\leq \Phi(R_0(B_t), R_1(B_t), \sigmageneric(\lS_t))\;, &\textit{(cost bound)}\label{eq:cost_bound_condition}\\
        \Phi(\bar x(s) + (\lambda,\mu,0)) &\leq \Phi(\bar x(s)) + O\left(\frac{\abs{\lambda}+\abs{\mu}}{\sigmageneric(V_\eta^1)}\right)\;, &\textit{(continuity)}\label{eq:continuity_condition} \\
        \Phi(\bar x(s)) &\leq \Phi(\bar x(s'))\;, &\textit{(monotonicity)}\label{eq:monotonicity_condition}
    \end{align}
    for all $t \in \{1,\ldots,k\}$, $B_t \subseteq \lS_t \cap V_{1/3}^1$, $s,s' \in [0, \sigmageneric(V_\eta^1)]$ with $s \leq s'$ and all $\lambda,\mu \in \bR$ such that the function arguments are inside the domain.
\end{definition}

\begin{lemma}\label{lemma:cost_bound_generic_cost_density}
    Let $\Phi$ be a valid cost density function and $\bar x$ as in \cref{def:cost-density-function}. Then
    \[\sum_{t=1}^k \cgeneric(S_t) \leq \int_0^{\sigmageneric(V_\eta^1)} \Phi(\bar x(s))ds + O(\epsilon^{1/3} )\cdot\Opt\;.\]
\end{lemma}
\begin{proof}
For $t \in \{1,\ldots,k\}$, let $B_t \subseteq \lS_t \cap V_{1/3}^1$ be a set satisfying the conditions of \cref{lemma:r0-r1-close-to-avg}. Together with the continuity assumption \eqref{eq:continuity_condition}, we obtain
\begin{equation}\label{eq:continuity_condition_applied}
    \Phi(R_0(B_t), R_1(B_t), \sigmageneric(\lS_t)) \leq \Phi(\bar x(\sigmageneric(\lS_t))) + O(\epsilon^{1/3})\;.
\end{equation}
The definition of $\bar D$ implies the cost bound
\begin{align*}
    \sum_{t=1}^k \cgeneric(S_t) &\leq \sum_{t=1}^k \sigmageneric(S_t)\cdot \phi_t + \norm{\bar D(k)}_1 \\
    &\stackalign{\eqref{eq:cost_bound_condition}}{\leq} \left(\sum_{t=1}^k \sigmageneric(S_t) \cdot \Phi(R_0(B_{t-1}), R_1(B_{t-1}), \sigmageneric(\lS_{t-1})) \right) + \norm{\bar D(k)}_1 \\
    &\stackalign{\eqref{eq:continuity_condition_applied}}{\leq} \left(\sum_{t=1}^k \sigmageneric(S_t)\cdot \Phi(\bar x(\sigmageneric(\lS_{t-1}))) \right) + O(\epsilon^{1/3}) \cdot \sigmageneric(V_\eta^1) + \norm{\bar D(k)}_1 \\
    &\stackalign{\cref{lemma:discrepancy_bound_final}}{\leq}\qquad \left(\sum_{t=1}^k \sigmageneric(S_t)\cdot \Phi(\bar x(\sigmageneric(\lS_{t-1}))) \right) + O(\epsilon^{1/3})\cdot \sigmageneric(V_\eta^1) \\
    &\stackalign{\eqref{eq:monotonicity_condition}}{\leq} \int_0^{\sigmageneric(V_\eta^1)} \Phi(\bar x(s))ds + O(\epsilon^{1/3}) \cdot \sigmageneric(V_\eta^1)\;.
\end{align*}
Using $\sigmageneric(V_\eta^1) \leq 2 \Opt$ finishes the proof. \qed
\end{proof}

\section{Cost Density Functions}\label{sec:cost-densities}
In this section, we derive valid cost density functions as defined in \cref{def:cost-density-function}. By \cref{lemma:cost_bound_generic_cost_density}, such a cost density immediately implies a cost bound for Algorithm \ref{alg:average_relative_greedy}. The main property we have to show is the bound on $\phi_t$ in \eqref{eq:cost_bound_condition}. The other two conditions will follow from a relatively straightforward computation in \cref{lemma:continuity_monotonicity_check}. We will derive two different cost density functions in \cref{lemma:cost_density_bound_simple} and \cref{lemma:cost_density_bound_complex}.

Recall that $\lS_t$ denotes the set of tours selected after $t$ iterations and $A_t \coloneq V_\eta^1\setminus \lS_t$ the available clients after $t$ iterations. In order to show a bound $\phi_{t+1} \leq \Phi \coloneq \Phi(R_0(B), R_1(B), \sigmageneric(\lS_t))$ for all $B \subseteq \lS_t \cap V_{1/3}^1$, we use the following approach: By \cref{lemma:greedy-rule}, we have $\phi_{t+1} \leq \cgeneric(\remainingtours[A_t]) / \sigmageneric(\remainingtours[A_t])$, where $\remainingtours \subseteq \lT^*$ is a subset of the optimum solution with $\remainingtours \supseteq \{T \in \lT^*\ |\ T \cap V_{1/3}^1 \subseteq A_t\}$. Therefore, it suffices to show
\begin{equation}\label{eq:rem-tours-density-non-positive}
    \cgeneric(\remainingtours[A_t]) - \Phi\cdot \sigmageneric(\remainingtours[A_t]) \leq 0\;.
\end{equation}
For a fixed set $B$, define the set of discarded tours as
\[\lT^*_\mathrm{dis} \coloneq \{T \in \lT^*\ |\ T \cap B \neq \emptyset\text{ and }\cgeneric(T) > \Phi \cdot \sigmageneric(T[A_t])\}\;.\]
Setting $\remainingtours = \lT^* \setminus \lT^*_\mathrm{dis}$, the left-hand side of \eqref{eq:rem-tours-density-non-positive} is bounded by
\begin{equation}\label{eq:phi-condition}
    \Opt - \Phi \cdot (\sigmageneric(V_\eta^1) - \sigmageneric(\lS_t)) - \sum_{\substack{T \in \lT^*\\ T \cap B \neq \emptyset}}\max\{0, c(T) - \Phi \cdot \sigmageneric(T[A_t])\}\;.
\end{equation}
Thus, it suffices to show that \eqref{eq:phi-condition} is non-positive.
\begin{remark}\label{remark:negative-denom-infinity}
For the statements of \cref{lemma:cost_density_bound_simple,lemma:cost_density_bound_complex}, we define fractions with a non-positive denominator to be infinity by convention. This simplifies notation since the denominator corresponds to a lower bound on $\sigmageneric(\remainingtours)$. If this lower bound is not positive, we do not obtain any bound on $\cgeneric(\remainingtours) / \sigmageneric(\remainingtours)$.
\end{remark}

\begin{lemma}\label{lemma:cost_density_bound_simple}
    Fix an iteration $t \in \{0,\ldots,k-1\}$ and let $B \subseteq \lS_t \cap V_{1/3}^1$. Then, with the convention from \cref{remark:negative-denom-infinity} and
    \[\Phi^{(1)}_\delta \coloneq \Phi^{(1)}_\delta\left(R_0(B), R_1(B), \sigmageneric(\lS_t)\right) \coloneq \frac{\Opt - R_0(B)}{\sigmageneric(V_\eta^1) - 2(R_0(B) - R_1(B)) - \sigmageneric(\lS_t)},\]
    we have $\phi_{t+1} \leq \max\{\frac12, \Phi^{(1)}_\delta\}$.
\end{lemma}
\begin{proof}
    As shown at the beginning of the section, it suffices to show that for $\Phi = \max\{\frac12, \Phi^{(1)}_\delta\}$, \eqref{eq:phi-condition} is non-positive. Setting $\hat d(v) = \min\{d(v), \tfrac12\}$, we show
    \begin{equation}\label{eq:bound_c-phisigma}
        \max\{0, c(T) - \Phi\cdot\sigmageneric(T[A_t])\} \geq\sum_{v \in T \cap B}\underbrace{\left(1- 2\Phi(1- \hat d(v))\right)}_{\eqcolon a_v}R_0(v)
    \end{equation}
    for any tour $T \in \lT^*$. Set $J \coloneq \{v \in T \cap B \mid a_v > 0\}$. If $J = \emptyset$, \eqref{eq:bound_c-phisigma} is immediate. Otherwise, using $\sigmageneric(T[A_t])\leq 2(1-\hat d(J))c(T)$ and $c(T)\geq R_0(v)$,
    \begin{align*}
        c(T) - \Phi\sigmageneric(T[A_t]) &\geq \left[1 - 2 \Phi (1-\hat d(J))\right]c(T)\\
        & = \left[\sum_{v \in J} a_v + (|J| -1)(2\Phi -1)\right]c(T) \geq \sum_{v \in T \cap B}a_vR_0(v)\;,
    \end{align*}
    where the last inequality uses $2\Phi -1 \geq 0$. This proves \eqref{eq:bound_c-phisigma}.
    Therefore, the term in \eqref{eq:phi-condition} is at most
    \begin{align*}
        \Opt - R_0(B) - \Phi \left[\sigmageneric(V_\eta^1) - \sigmageneric(\lS_t) - 2(R_0(B)-R_1(B))\right],
    \end{align*}
    By definition of $\Phi$, this is at most zero, which finishes the proof. \qed
\end{proof}

\begin{lemma}\label{lemma:cost_density_bound_complex}
    Fix an iteration $t \in \{0,\ldots,k-1\}$ and let $B \subseteq \lS_t \cap V_{1/3}^1$. Then, with the convention from \cref{remark:negative-denom-infinity}, $\delta = 1/3$, and
    \begin{align*}
        \Phi^{(2)}_{1/3} 
        \coloneq \frac{\Opt - R_0(B)}{\sigma_{1/3}(V_\eta^1) - \frac32 (R_0(B) - R_1(B))- \sigma_{1/3}(\lS_t) - \frac18 R_0(V_\double)}\;,
    \end{align*}
    we have $\phi_{t+1} \leq \max\{\frac47, \Phi^{(2)}_{1/3}\}$. As before, $\Phi^{(2)}_{1/3}$ can be viewed as a function of $R_0(B)$, $R_1(B)$ and $\sigmageneric(\lS_t)$.
\end{lemma}
\begin{proof}
    We proceed similarly to \cref{lemma:cost_density_bound_simple} and show that \eqref{eq:phi-condition} is non-positive for $\Phi = \max\{\frac47, \Phi^{(2)}_{1/3}\}$. Set again $\hat d(v) \coloneq \min\{d(v), \tfrac12\}$. We claim
    \begin{equation}\label{eq:bound-c-phi-sigma-third}
    \begin{aligned}
    \max\{0, c(T) &- \Phi \cdot \sigmathird(T[A_t])\} \\
        \geq &\sum_{v \in T\cap B}\Big(\underbrace{ 1- \tfrac32 \Phi \left(1 - \hat d(v)\right)}_{\eqcolon a_v}\Big) R_0(v)
        - \tfrac{\Phi}{8}R_0(T \cap V_\double)\;.
    \end{aligned}
    \end{equation}
    If none of the coefficients $a_v$ is positive, the inequality is trivial. If $T \in \lT^*_\single$ and $T \cap B = \{v\}$, then
    \[\sigmathird(T[A_t]) \leq \tfrac32 \bigl(1-\hat d(v)\bigr)c(T)\;,\]
    which, together with $c(T) \geq R_0(v)$, implies \eqref{eq:bound-c-phi-sigma-third}. It remains to show the inequality for $T \in \lT^*_\double$, i.e.\ $T \cap V_{1/3}^1 = \{v,w\}$. If $\abs{T \cap B}=1$, say $T \cap B = \{v\}$, then
    \begin{align*}
        \sigmathird(T[A_t]) \leq \tfrac32 \left(1-\hat d(v) - \hat d(w)\right)c(T)
         + \left(3\hat d(w) - \tfrac12\right)R_0(w)
    \end{align*}
    If $a_v - \frac{\Phi}{8} \leq 0$, then \eqref{eq:bound-c-phi-sigma-third} is trivial. Otherwise,
    \begin{align*}
        &c(T) - \Phi \cdot \sigmathird(T[A_t]) \\
        &\geq\ \left(1-\tfrac32\Phi\bigl(1-\hat d(v) - \hat d(w)\bigr)\right)c(T)  - \Phi\left(3\hat d(w) - \tfrac12\right)R_0(w) \\
        &=\ \left(a_v + \tfrac32 \Phi \cdot\hat d(w)\right)c(T) - \Phi\left(3\hat d(w) - \tfrac12\right)R_0(w) \\
        &\geq\ \left(a_v - \tfrac{\Phi}{8}\right)c(T) - \Phi\left(\tfrac32 \hat d(w) - \tfrac58\right)R_0(w) \\
        &\geq\ \left(a_v - \tfrac{\Phi}{8}\right)c(T) - \tfrac{\Phi}{8} R_0(w) \\
        &\geq\ a_v R_0(v) - \tfrac{\Phi}{8}R_0(T \cap V_\double)\;,
    \end{align*}
    which shows \eqref{eq:bound-c-phi-sigma-third}. If $T \cap B = \{v,w\}$, then
    \begin{equation}\label{eq:bound-savings-complex-tour-twice}
        \max\{0, c(T) - \Phi \cdot \sigmathird(T[A_t])\} \geq \left(1 - \tfrac32\Phi(1 - \hat d(v) - \hat d(w))\right) c(T)\;.
    \end{equation}
    Consider the coefficients $a_v - \frac{\Phi}{8}$ and $a_w - \frac{\Phi}{8}$ of $R_0(\cdot)$ on the right-hand side of \eqref{eq:bound-c-phi-sigma-third}. If neither term is positive, \eqref{eq:bound-c-phi-sigma-third} is obvious. If exactly one is positive, say $a_v - \frac{\Phi}{8}$, then it follows from \eqref{eq:bound-savings-complex-tour-twice} by omitting $\hat d(w)$ and replacing $c(T)$ by $R_0(v)$. If both terms are positive, then \eqref{eq:bound-c-phi-sigma-third} can be seen by comparing the sum of coefficients
    \begin{align*}
        &1 - \tfrac32 \Phi \bigl(1- \hat d(v) - \hat d(w)\bigr) \\
        =\ &\left(1 - \tfrac32\Phi\bigl(1 - \hat d(v)\bigr) - \tfrac{\Phi}{8}\right) + \left(1 - \tfrac32\Phi\bigl(1 - \hat d(w)\bigr) - \tfrac{\Phi}{8}\right) + \tfrac74\Phi - 1
    \end{align*}
    where $\tfrac74\Phi -1 \geq 0$ by definition, and then replacing $c(T)$ by $R_0(v)$ and $R_0(w)$.

    Finally, by inserting \eqref{eq:bound-c-phi-sigma-third}, we see that \eqref{eq:phi-condition} is at most
    \[\Opt - R_0(B) - \Phi \left[\sigmathird(V_\eta^1) - \sigmathird(\lS_t)  - \tfrac32 (R_0(B) - R_1(B)) - \tfrac{1}{8}R_0(V_\double)\right]\;,\]
    which finishes the proof by definition of $\Phi^{(2)}_{1/3}$ and $\Phi = \max\{\frac47, \Phi^{(2)}_{1/3}\}$. \qed
\end{proof}

\begin{lemma}\label{lemma:continuity_monotonicity_check}
    The functions
    \[\min\left\{\max\left\{\tfrac12, \Phi^{(1)}_\eta\right\}, 1\right\}, \qquad \min\left\{\max\left\{\tfrac12, \Phi^{(1)}_{1/3}\right\}, \max\left\{\tfrac47, \Phi^{(2)}_{1/3}\right\}, 1\right\}\]
    satisfy the requirements \eqref{eq:continuity_condition} and \eqref{eq:monotonicity_condition} for a cost density function from \cref{def:cost-density-function}.
\end{lemma}
This statement follows from elementary computations and is proven in \cref{sec:proof_continuity_monotonicity_check}.

\begin{corollary}\label{cor:eta-greedy-final-cost-bound}
    The cost of the solution obtained by Algorithm \ref{alg:average_relative_greedy} with $\delta=\eta$ is at most $C_\eta$ with
    \[C_\eta \coloneq \alpha \cdot \Opt + \sigmaeta(V_0^\eta) + \int_0^{\sigmaeta(V_\eta^1)} \min\{\Phi_\eta^{(1)}(\bar x(s)), 1\}ds\; + O(\epsilon^{1/3})\Opt\;.\]
\end{corollary}
\begin{proof}
    By \cref{lem:delta_tank}, the cost of the solution produced by the algorithm is at most
    \[\alpha \cdot \Opt + \sigmaeta(V_0^\eta) + \sum_{T \in \lS_k}c_\eta(T)\;.\]
    By \cref{lemma:cost_density_bound_simple} and the definition of $c_\eta$, the function $\min\{\max\{1/2, \Phi_\eta^{(1)}\}, 1\}$ satisfies the cost bound \eqref{eq:cost_bound_condition}. The other two requirements for a valid cost density function are checked in \cref{lemma:continuity_monotonicity_check}. By monotonicity and $\sigmaeta(V_\eta^1) \leq 2\Opt$, $\Phi^{(1)}_\eta(\bar x(s)) \geq \tfrac12$, so the maximum with $\tfrac12$ can be omitted. The result therefore follows from \cref{lemma:cost_bound_generic_cost_density}. \qed
\end{proof}

\begin{corollary}\label{cor:third-greedy-final-cost-bound}
    The cost of the solution obtained by Algorithm \ref{alg:average_relative_greedy} with $\delta=\frac13$ is at most $C_{1/3}$ with
    \begin{align*}
    C_{1/3} \coloneq \alpha \cdot \Opt + \sigmathird(V_0^\eta) + \int_0^{\sigmathird(V_\eta^1)}\min\{\Phi_{1/3}^{(1)}(\bar x(s)), \Phi_{1/3}^{(2)}(\bar x(s)), 1\}ds\\ + O(\epsilon^{1/3})\Opt\;.
    \end{align*}
\end{corollary}
\begin{proof}
    We first note that $\Phi_{1/3}^{(2)}(\bar x(s)) \geq \tfrac47$: By \cref{lemma:continuity_monotonicity_check} it is monotonously increasing in $s$. At $s = 0$, we use that
    \begin{align*}
        \sigmathird(V_\eta^1) - \tfrac18 R_0(V_\double) \leq \tfrac{3}{2}\Opt + \tfrac14 R_0(V_\single) + \tfrac18 R_0(V_\double) \overset{\eqref{eq:single-double-radial}}{\leq} \tfrac 74 \Opt.
    \end{align*}
    We now proceed analogously to \cref{cor:eta-greedy-final-cost-bound}, noting also that $\Phi_{1/3}^{(1)}(\bar x(s)) \geq \tfrac12$, so both maxima are redundant. The necessary cost bounds are shown in \cref{lemma:cost_density_bound_simple} and \cref{lemma:cost_density_bound_complex}, and continuity and monotonicity are checked in \cref{lemma:continuity_monotonicity_check}. The result then follows from \cref{lemma:cost_bound_generic_cost_density}. \qed
\end{proof}

\section{Deriving the Approximation Ratio}\label{sec:compute-apx-ratio}

The objective of this section is to prove \cref{thm:main_theorem}. To this end, we analyze the algorithm that runs the Matching Savings Algorithm and the Average Greedy Algorithm with each of the potentials $\sigmaeta$ and $\sigmathird$, then returns the cheapest of the three solutions. We show that, for every CVRP instance, at least one of the three cost bounds is at most $(\alpha + 1.659)\Opt$, where $\alpha$ is the approximation ratio for TSP.

After simplifying the cost bounds, we will see that they depend on the same quantities
\begin{equation}\label{eq:variables}
    R_1(V_0^\eta), R_1(V_\eta^{1/3}), R_1(V_\single), R_1(V_\double), R_0(V_\single), R_0(V_\double)\;.
\end{equation}
We show that, for every feasible choice of these quantities, at least one of the three bounds is sufficiently small. Throughout this section, we assume that the instance is scaled such that $\Opt=1$.

\begin{lemma}\label{thm:final_approx_ratio}
    Let $C_M$ be the cost bound for the matching algorithm (Algorithm \ref{alg:matching}) from \cref{lem:MatchingSavingsCostBound}, and let $C_\eta,\ C_{1/3}$ denote the cost bounds for the Average Greedy (Algorithm \ref{alg:average_relative_greedy}) from \cref{cor:eta-greedy-final-cost-bound} and \cref{cor:third-greedy-final-cost-bound}. For sufficiently small constants $\epsilon, \eta > 0$, every CVRP instance $I$ satisfies 
    \[\min \{C_M(I), C_\eta(I), C_{1/3}(I)\} \leq \bigl(\alpha + 1.659 \bigr) \Opt.\]
\end{lemma}

We first use \cref{thm:final_approx_ratio} to finish the proof of \cref{thm:main_theorem}.
\begin{proofof}{\cref{thm:main_theorem}}
    Fix sufficiently small $\epsilon, \eta > 0$ as in \cref{thm:final_approx_ratio}. Applying \cref{thm:final_approx_ratio} proves that returning the cheapest of the three solutions gives the claimed approximation guarantee.

    It remains to prove a polynomial runtime. The Matching Savings Algorithm runs in polynomial time. 
    For the Average Greedy Algorithm, every feasible tour in $V_\eta^1$ contains at most $1/\eta$ clients. Thus, there are $n^{O(1/\eta)}$ candidate tours, where $n=|V|$, and at most $n$ iterations. Hence, examining these
    candidates in each iteration can be done in polynomial time. 
    For fixed $\epsilon$ and $\eta$, the reduction in \cref{app:proof-any-tour-small-r1} uses only a constant number
    of copies, and we consider only a constant number of guess vectors $(\tilde r_b)_{b=1}^{\lceil1/\epsbucket\rceil}$ (see \cref{sec:greedy}).
    Hence, the combined algorithm runs in polynomial time.
\end{proofof}
It remains to prove \cref{thm:final_approx_ratio}. The proof proceeds in three steps. 
First we simplify the cost bound for the Average Greedy with potential $\sigmaeta$ and omit the error terms in $C_\eta$ and $C_{1/3}$. Next, we use the Matching Savings bound to restrict to the configurations of $R_0$- and $R_1$-values that give bad cost bounds for that algorithm. Doing so we reduce the six-dimensional problem to a three-dimensional prism. Finally, we use interval arithmetic to certify an upper bound over this prism. AI tools were used to assist with the computations in \cref{lem:three-dimensional-reduction} and the implementation of the interval-arithmetic verification.

We first simplify the bound for the Average Greedy with potential $\sigmaeta$. For every $W \subseteq V$,
\[\sigma_\eta(W) \leq  (1 + O(\eta)) R_1(W_0^\eta) + 2 R_1(W_\eta^1). \]
Replacing $\sigmaeta(V_\eta^1)$ by its upper bound $2R_1(V_\eta^1)$ only increases the cost bound $C_\eta$ in \cref{cor:eta-greedy-final-cost-bound}:
Substituting $s=t\,\sigma_\eta(V_\eta^1)$ rescales the integration interval to $[0,1]$. A short differentiation argument then shows that the resulting integral, including its prefactor, is non-decreasing in $\sigmaeta(V_\eta^1)$. Define
\[
    \begin{aligned}
       \widehat C_\eta \coloneq &\alpha \cdot \Opt+R_1(V_0^\eta)\\
        + &\int_0^{2R_1(V_\eta^1)}\min\left\{1,\,
        \frac{ \Opt- \frac{R_0(V_{\mathrm{single}})}{2 R_1(V_\eta^1)}s}
        {2R_1(V_\eta^1) - \left( 1+ \dfrac{ 2R_0(V_{\mathrm{single}})-2R_1(V_{\mathrm{single}})}{2 R_1(V_\eta^1)}\right)s  }
        \right\}\,ds \;.     
    \end{aligned}
\]
\begin{remark}
The analysis in \cref{sec:greedy} and \cref{sec:cost-densities} does not require the explicit formula for $\sigmaeta$. One can check that the proof of the cost bound also works when using $2R_1(v)$ as the potential for $v \in V_\eta^1$. Consequently, running the Average Greedy  with this modified potential yields the same cost bound.
\end{remark}
Similarly, denote by $\widehat C_{1/3}$ the cost bound obtained by omitting the error term in \cref{cor:third-greedy-final-cost-bound}.
 The resulting term is
\[
    \hat C_{1/3} = \alpha\cdot \Opt + \frac32 R_1(V_0^\eta) + \int_0^{\sigmathird(V_\eta^1)}\min\{\Phi^{(1)}_{1/3}(\bar x(s)), \Phi^{(2)}_{1/3}(\bar x(s)), 1\}ds\;,
\]
where inserting $\bar x(s)$ from \cref{def:cost-density-function} into the cost density functions from \cref{lemma:cost_density_bound_simple,lemma:cost_density_bound_complex} and setting $t \coloneq s / \sigmathird(V_\eta^1)$ yields
\[
\begin{aligned}
    \Phi^{(1)}(\bar x(s)) &= \frac{\Opt - t \cdot R_0(V_\single)}{(1-t)\sigmathird(V_\eta^1) - 2t(R_0(V_\single) - R_1(V_\single))}, \\[0.2cm]
    \Phi^{(2)}(\bar x(s)) &= \frac{\Opt - t \cdot R_0(V_\single)}{(1-t)\sigmathird(V_\eta^1) - \frac32 t(R_0(V_\single) - R_1(V_\single)) - \frac18 R_0(V_\double)}\;.
\end{aligned}
\]

Here, we again use the convention from \cref{remark:negative-denom-infinity} to set the fractions to infinity whenever the denominator is non-positive.
The preceding analysis shows $ C_\delta \leq\widehat C_\delta+O(\epsilon^{1/3}+\eta)\Opt$, for $\delta\in\{\eta,1/3\}$.
Set $\rho \coloneq \alpha + 1.659$, normalize $\Opt =1$ and consider the compact region of the quantities in \eqref{eq:variables} defined by non-negativity and
\begin{equation}
\begin{aligned}\label{eq:cvrp-instance-constraints}
    R_1(V) \leq 1, \quad R_0(V_\single) + \frac12 R_0(V_\double) \leq 1,\\
    2R_1(V_\single) \leq R_0(V_\single) \leq 3R_1(V_\single),\\
    2R_1(V_\double) \leq R_0(V_\double) \leq 3R_1(V_\double)
\end{aligned}
\end{equation}
This region contains the quantities associated with every normalized instance by definitions of $R_0$ and $R_1$ and inequality \eqref{eq:single-double-radial}. It suffices to prove
\[\min\{C_M, \widehat C_\eta, \widehat C_{1/3}\} < \rho\]
throughout this region. Indeed, the cost bounds define continuous functions on a compact region and hence the minimum over the three cost bounds attains its maximum on the compact region. This gives a constant $\kappa > 0$ such that for every normalized CVRP instance $I$
\[\min\{C_M, \widehat C_\eta, \widehat C_{1/3}\} < \rho - \kappa\]
For small enough $\epsilon, \eta > 0$, the uniform error bound $O(\epsilon^{1/3} + \eta) \Opt$ can then be absorbed into that slack, proving \cref{thm:final_approx_ratio}.
\begin{lemma}[Reduction to a three-dimensional prism]
\label{lem:three-dimensional-reduction}
Suppose that 
\begin{equation}
\label{eq:three-dimensional-certificate}
    \min\bigl\{\widehat C_\eta(r, m, \lambda),
                    \widehat C_{1/3}(r, m, \lambda)\bigr\}
    < \rho
\end{equation}
on the triangular prism
\begin{equation}
\label{eq:three-dimensional-prism}
    0.659 \leq r\leq 1,
    \qquad
    m\geq0,
    \qquad
    \lambda \geq0,
    \qquad
    m + \lambda \leq\tfrac{91}{750},
\end{equation}
where the cost bounds are evaluated with
\begin{equation}
\label{eq:variable-transformation-cvrp-inst}
\begin{aligned}
 R_1(V_0^\eta)     &=1-r,
&R_1(V_\eta^{1/3}) &=r-\tfrac{841}{1500}+\lambda,\\
 R_1(V_\single)    &=\tfrac{159}{500}+m+\lambda,
&R_1(V_\double)    &=\tfrac{91}{375}-m-2\lambda,\\
 R_0(V_\single)    &=\tfrac{159}{250}+2m+3\lambda,
&R_0(V_\double)    &=\tfrac{91}{125}-4m-6\lambda.
\end{aligned}
\end{equation}
Then
\[\min\{C_M, \widehat C_\eta, \widehat C_{1/3}\} < \rho \]
throughout the compact region defined by \eqref{eq:cvrp-instance-constraints}.
\end{lemma}
\begin{proof}
Suppose for contradiction that all three bounds are at least $\rho$. We regard the cost bounds as functions in the quantities \eqref{eq:variables} under the constraints \eqref{eq:cvrp-instance-constraints}. Intermediate tuples need not correspond to actual CVRP instances. Increasing $R_1(V_0^\eta)$ until $R_1(V) = \Opt = 1$ increases both greedy bounds and leaves $C_M$ fixed. Thus, we may assume $R_1(V) = 1$. Define
\begin{align*}
    m &\coloneq \tfrac12\Bigl(3R_1(V_\single)-R_0(V_\single)+3R_1(V_\double)-R_0(V_\double)\Bigr)-0.159\;, \\
    \lambda &\coloneq R_0(V_\single) - 2 R_1(V_\single)\;, \\
    \mu &\coloneq 2 - 2R_0(V_\single) - R_0(V_\double)\;, \\
    r &\coloneq R_1(V_\eta^{1/3}) + R_1(V_\single) + R_1(V_\double)\;, \\
    z &\coloneq 3R_1(V_\double) - R_0(V_\double)
\end{align*}
Using $R_1(V) = 1$, one can verify that this defines an invertible affine transformation from the original variables.
Since $\rho \leq C_M = \rho + m - z$, the region bounds imply
\begin{equation}\label{eq:z-leq-m}
    0 \leq z \leq m \qquad \text{and} \qquad \lambda, \mu \geq 0\;.
\end{equation}
Moreover, since $C_M \geq \rho$
\[3R_1(V_\single) - R_0(V_\single) -z \geq \tfrac{159}{500}.\]
Rearranging gives $R_1(V_\single) \geq 1/3(R_0(V_\single)+z) + \frac{53}{500}$. Using this bound in the first inequality and using $z\leq\frac12R_0(V_\double)$ and \eqref{eq:single-double-radial} in the second inequality, we obtain
\begin{align*}
    \tfrac{159}{500}+2m+2\lambda
    &=R_0(V_\single)-R_1(V_\single)+z\\
    &\leq\tfrac23\bigl(R_0(V_\single)+z\bigr)-\tfrac{53}{500}
    \leq\tfrac{841}{1500}.
\end{align*}
Furthermore, $\rho \leq \widehat C_\eta \leq \alpha + (1-r) + 2r$ since the integrand in $\widehat C_\eta$ is at most one. Hence,
\begin{equation}\label{eq:m-lambda-r-bounds}
m+\lambda \leq \tfrac{91}{750}, \qquad 0.659 \leq r \leq 1\;.
\end{equation}

In the following, we will eliminate $\mu$ and $z$, while maintaining \eqref{eq:z-leq-m} and \eqref{eq:m-lambda-r-bounds} and without decreasing either greedy bound. Observe that 
\begin{equation}\label{eq:sigma-third-in-new-variables}
\begin{aligned}
    \sigmathird(V_\eta^1) 
    = \tfrac32 R_1(V_\eta^{1/3}) &+ 3(R_1(V_\single) + R_1(V_\double))\\
     &- \tfrac12 (R_0(V_\single) + R_0(V_\double)) 
    = \tfrac32 r + m + 0.159\;.
\end{aligned}
\end{equation}
We rescale the integrals to $t\in[0,1]$ by setting $\sigmaeta(V_\eta^1) t = s$ and $\sigmathird(V_\eta^1) t = s$, respectively. For the remainder of the proof, $\sigmaeta(V_\eta^1)$ and $\sigmathird(V_\eta^1)$, which only depend on $r$ and $m$, will not change, so it suffices to consider the scaled integrands. Their common numerator is
\[
    N(t):=1-R_0(V_\single)t,
\]
and their denominators are
\begin{align*}
 D_\eta(t)&\coloneq(1-t)2r
              -2\bigl(R_0(V_\single)-R_1(V_\single)\bigr)t,\\
 D_1(t)&\coloneq(1-t)\sigmathird(V_\eta^1)
              -2\bigl(R_0(V_\single)-R_1(V_\single)\bigr)t,\\
 D_2(t)&\coloneq (1-t)\sigmathird(V_\eta^1)
              -\tfrac32\bigl(R_0(V_\single)-R_1(V_\single)\bigr)t-\tfrac18R_0(V_\double).
\end{align*}
The integrands are
\[
    \min\left\{1,\frac{N}{D_\eta}\right\},
    \qquad
    \min\left\{1,\frac{N}{D_1},\frac{N}{D_2}\right\}.
\]
First, set $\mu$ to zero while keeping $r,m,\lambda$ and $z$ fixed. This corresponds to the following changes:
\[
\begin{aligned}
 R_0(V_\double)&\gets R_0(V_\double)+\mu,\\
 R_1(V_\double)&\gets R_1(V_\double)+\tfrac{\mu}{3},\\
 R_1(V_\eta^{1/3})&\gets R_1(V_\eta^{1/3})-\tfrac{\mu}{3}.
\end{aligned}
\]
Note that also $\sigmathird(V_\eta^1)$ remains unchanged by \eqref{eq:sigma-third-in-new-variables}. Therefore, the expressions $N,D_\eta,D_1$ remain unchanged, while $D_2$ decreases by $\mu/8$.
This however can only increase the integrand. Thus, neither greedy bound decreases.
Now using $\mu=0$, one can verify that the original problem variables are given by
\begin{equation}
\label{eq:reverse-transformation}
\begin{aligned}
 R_1(V_0^\eta)&=1-r,\\
 R_1(V_\eta^{1/3})
   &=r-\tfrac{841}{1500}+\lambda+\tfrac23(m-z),\\
 R_1(V_\single)
   &=\tfrac{159}{500}+2m-z+\lambda,\\
 R_1(V_\double)
   &=\tfrac{91}{375}-\tfrac83m+\tfrac53z-2\lambda,\\
 R_0(V_\single)
   &=\tfrac{159}{250}+4m-2z+3\lambda,\\
 R_0(V_\double)
   &=\tfrac{91}{125}-8m+4z-6\lambda.
\end{aligned}
\end{equation}

Next, increase $z$ to $m$ while keeping $r,m,\lambda$ fixed. Differentiating the numerator and denominators with the definitions from \eqref{eq:reverse-transformation} gives
\[
    \partial_zN=\partial_zD_\eta=\partial_zD_1=2t,
    \qquad
    \partial_zD_2=-\tfrac12+\tfrac32t.
\]
In particular, $N(t)>0$ for $t<1$ throughout this transformation.
Whenever $0<N\leq D_i$, differentiation yields
\begin{align*}
 \partial_z\frac{N}{D_i}
    &=\frac{2t(D_i-N)}{D_i^2}\geq0,
      &&i\in\{\eta,1\},\\
 \partial_z\frac{N}{D_2}
    &=\frac{\tfrac12(1+t)N+2t(D_2-N)}{D_2^2}\geq0.
\end{align*}
Hence, both greedy bounds have not decreased. At $z = m$, the relations in \eqref{eq:reverse-transformation} give \eqref{eq:variable-transformation-cvrp-inst}, with $(r, m, \lambda)$ in the prism by \eqref{eq:m-lambda-r-bounds}. Both Greedy bounds are at least $\rho$ contradicting \eqref{eq:three-dimensional-certificate}. \qed
\end{proof}

Using interval arithmetic, we certify that the maximum of $\min\{\widehat C_\eta,\widehat C_{1/3}\}$ over the prism
in \eqref{eq:three-dimensional-prism} is strictly less than $\rho=\alpha+1.659$.
The interval-arithmetic verification code is available on
\href{https://github.com/Weismantel/Average_Greedy_for_CVRP}{GitHub}.
Together with the compactness argument this completes the proof of \cref{thm:final_approx_ratio} and therefore also of \cref{thm:main_theorem}.

\begin{remark}
The bound of $\alpha + 1.659$ is almost tight for our analysis.
For an instance with
\begin{equation*}
\begin{aligned}
    R_1(V_0^\eta) &= 0.174, &R_1(V_\eta^{1/3}) &= 0.267, \\
    R_1(V_\single) &= 0.437, &R_1(V_\double) &= 0.122, \\
    R_0(V_\single) &= 0.875, &R_0(V_\double) &= 0.248,
\end{aligned}
\end{equation*}
the bounds $C_M, \widehat C_\eta, \widehat C_{1/3}$ are almost identical and the minimum is larger than $\alpha + 1.6588$.
\end{remark}

\printbibliography

\newpage
\appendix
% This is necessary to fix links into the appendix. Without this, the
% appendix sections get the same internal identifiers as the sections in the document.
\renewcommand{\theHsection}{appendix.\thesection}

\section{Proof of \cref{lem:delta_tank}}\label{app:proof-delta-itp}

\begin{proofof}{\cref{lem:delta_tank}}
First serve all clients with demand greater than $1/2$ by singleton tours with cost exactly $\sigmageneric(V_{1/2}^1)$. The rest of the algorithm proceeds in two steps. First, compute a TSP tour $Q = (s, v_1,\ldots, v_n, s)$ with $c(Q) \leq \alpha \cdot \Opt$. Afterward, split the TSP tour into segments, that will be connected to the depot to form feasible tours. To this end, let $\tau \in [0,1)$ be an offset chosen uniformly at random. Let $I \subseteq \{1,\ldots,n\}$ be the set of indices $i$ for which
\[\Big\lfloor \tau + \frac{1}{1-\delta}\sum_{j = 1}^{i-1} d(v_j)\Big\rfloor \neq \Big\lfloor \tau + \frac{1}{1-\delta}\sum_{j = 1}^{i} d(v_j)\Big\rfloor.\]
For any consecutive indices $i_1, i_2$ it holds  that $d(\{v_{i_1 + 1}, \ldots, v_{i_2 - 1}\}) \leq 1 - \delta$. Furthermore, $ \Pr[i \in I] \leq d(v_i)/(1-\delta)$. Now split the TSP tour at these indices. If $v_i$ for $i \in I$ fits into the previous tour segment, assign it to this segment. Otherwise, create a singleton segment. Each segment is then connected to the depot. The additional connection cost is charged to the client where the tour was split. A client $v_i$ with $d(v_i) \leq \delta$ always fits into the previous tour segment, therefore the expected connection cost for these clients is
    \[
    \sum_{v_i \in V_0^\delta} \Pr[i \in I] \cdot R_0(v_i)\leq  \sum_{v_i \in V_0^\delta} \frac{d(v_i)}{1-\delta} R_0(v_i) = \frac{1}{1-\delta}R_1(V_0^\delta)\;.
    \]
A client $v_i \in V_\delta^{1/2}$, might not always fit into the previous tour segment. However, $\Pr[v_i \text{ fits} \mid i \in I] \geq \frac{\delta}{d(v_i)}$. Therefore, the expected connection cost for $v_i \in V_\delta^{1/2}$ is
\begin{align*}
    \Pr[i \in I] \cdot &\bigl(\Pr[v_i \text{ fits } | i \in I] \cdot 1  + \Pr[v_i \text{ does not fit } | i \in I] \cdot 2\bigr) R_0(v_i)\\
    &\leq \frac{d(v_i)}{1-\delta} \left( \bigl(\tfrac{\delta}{d(v_i)}\bigr) \cdot 1 + \bigl(1 - \tfrac{\delta}{d(v_i)}\bigr) \cdot 2 \right) R_0(v_i) = \frac{2d(v_i)-\delta}{1-\delta}R_0(v_i).
\end{align*}
Summing over all clients in $V_\delta^{1/2}$ yields:
    \[
    \sum_{v \in V_\delta^1} \frac{2d(v)-\delta}{1-\delta} R_0(v) = \frac{2}{1-\delta}R_1(V_\delta^1) - \frac{\delta}{1-\delta}R_0(V_\delta^1).
    \]
Combining the TSP cost, the splitting cost and the singleton tour cost yields the claimed bound. Note that the tour splitting can be solved optimally by dynamic programming. Therefore, the cost bound also holds deterministically.
\end{proofof}

\section{Reduction to Small Tours}\label{app:proof-any-tour-small-r1}
\begin{lemma}\label{lem:any-tour-small-r1}
    Let $\epsilon, \eta>0$ and $\rho \geq 1$. If there is a $\rho$-approximation algorithm for CVRP with the additional requirement that
    $R_1(T) \leq \epsilon\cdot R_1(V_\eta^1)$ for any tour $T$ on $V_\eta^1$, then there is a $\rho$-approximation for CVRP.
\end{lemma}
\begin{proof}
    Given a CVRP instance $I=(V,s,c,d)$, create an instance $I'$ by copying the original instance $C=\ceil{\frac{1}{\epsilon\cdot\eta}}$ times and identifying the depot vertices. The metric between clients of different copies is given by the metric closure, i.e.\ by taking a detour through the depot. Given a solution to $I$, one can obtain a solution to $I'$ with at most $C$ times the cost by copying each tour $C$ times. Conversely, a solution to $I'$ can be turned into $C$ solutions to $I$ with in total at most the same cost: By definition of the metric, one can subdivide each tour whenever it crosses between different copies of the original instance without increasing the cost. In particular, it suffices to compute a $\rho$-approximate solution to $I'$.

    It remains to check that $I'$ satisfies the condition. If $V_\eta^1 = \emptyset$, this is clear.
    Otherwise, let $M = \max_{v \in V_\eta^1} R_0(v)$ and $T$ be an arbitrary tour in $(V_\eta^1)^C$. Then
    \[R_1(T) \leq \sum_{v \in T} d(v) R_0(v) \leq M \leq \epsilon C \eta M \leq \epsilon C R_1(V_\eta^1) =\epsilon R_1\bigl((V_\eta^1)^C\bigr).\]
\end{proof}

\section{Proof of \cref{lemma:continuity_monotonicity_check}}\label{sec:proof_continuity_monotonicity_check}
\begin{proofof}{\cref{lemma:continuity_monotonicity_check}}
    Recall the definition
    \[\bar x(s) \coloneq \left(\frac{R_0(V_\single)}{\sigmageneric(V_\eta^1)}\cdot s, \frac{R_1(V_\single)}{\sigmageneric(V_\eta^1)}\cdot s, s\right) \in \bR^3_{\geq0}\]
    and the continuity condition
    \[\Phi(\bar x(s) + (\lambda,\mu,0)) \leq \Phi(\bar x(s)) + O\left(\frac{\abs{\lambda}+\abs{\mu}}{\sigmageneric(V_\eta^1)}\right)\]
    from \cref{def:cost-density-function}. We will show this for the functions $\Phi=\min\{\max\{1/2, \Phi^{(1)}_\eta\}, 1\}$ and $\Phi = \min\{\max\{1/2, \Phi^{(1)}_{1/3}\}, \max\{4/7, \Phi^{(2)}_{1/3}\}, 1\}$. Clearly, the continuity is preserved by taking minima and maxima, so it suffices to show it for the fractions on the domain where they attain the minimum.

    Let $\delta \in \{\eta, 1/3\}, s \in [0, \sigmageneric(V_\eta^1)]$ and $t \coloneq \frac{s}{\sigmageneric(V_\eta^1)} \in [0,1]$. Additionally, let $\lambda,\mu \in \bR$ such that $\bar x(s) + (\lambda,\mu,0)$ is inside the domain. We have
    \begin{equation}\label{eq:condition-lemma-cost-density-1}
        \Phi_\delta^{(1)}(\bar x(s))
        = \frac{\Opt - t\cdot R_0(V_\single)}{(1-t)\sigmageneric(V_\eta^1) - 2t\cdot (R_0(V_\single) - R_1(V_\single))}
    \end{equation}
    and similarly, with $\delta=\frac13$,
    \begin{equation}\label{eq:condition-lemma-cost-density-2}
        \Phi_\delta^{(2)}(\bar x(s))
        = \frac{\Opt - t\cdot R_0(V_\single)}{(1-t)\sigmageneric(V_\eta^1) - \frac32 t\cdot (R_0(V_\single) - R_1(V_\single)) - \frac18 R_0(V_\double)}\;.
    \end{equation}
    Since $R_1(V_\single) \leq \frac12 R_0(V_\single)$, both terms are at least
    \begin{align*}
        \frac{\Opt - t\cdot R_0(V_\single)}{(1-t)\sigmageneric(V_\eta^1) - \frac{3}{4}t\cdot R_0(V_\single)}\;.
    \end{align*}
    If $\Phi_\delta^{(i)}(\bar x(s)) > 1$ for $i \in \{1,2\}$, then the minimum is attained by 1 and the continuity requirement is trivially satisfied. Otherwise, we have
    \begin{equation}\label{eq:tbound}\Opt - \frac{t}{4} \cdot R_0(V_\single) \leq (1-t) \sigmageneric(V_\eta^1)\end{equation}
    and therefore
    \[t \leq \frac{\sigmageneric(V_\eta^1) - \Opt}{\sigmageneric(V_\eta^1) - \frac14 R_0(V_\single)} \leq \frac{\sigmageneric(V_\eta^1) - \Opt}{\sigmageneric(V_\eta^1) - \frac14 \Opt} \leq \frac{\Opt}{\frac74 \Opt}=\frac47\;,\]
    where we used $R_0(V_\single) \leq \Opt$ and $\sigmageneric(V_\eta^1) \leq 2\Opt$. Note that $\Opt \leq \sigmageneric(V_\eta^1)$ by \eqref{eq:tbound} since otherwise the minimum is always attained by 1, so none of the fractions above are negative.
    
    Let $\Phi_\delta^{(i)}(\bar x(s))=\frac{N}{D}$, where $N$ and $D$ are the numerator and denominator of \eqref{eq:condition-lemma-cost-density-1} or \eqref{eq:condition-lemma-cost-density-2}. Then $t \leq \frac47$ implies  $N \geq \frac37 \Opt \geq \frac3{14} \sigmageneric(V_\eta^1)$, so also $D \geq \frac3{14} \sigmageneric(V_\eta^1)$.

    We can assume $\abs{\lambda} + \abs{\mu} \leq (1/24)\cdot \sigmageneric(V_\eta^1)$: Otherwise, the continuity condition is trivial because $\Phi \in [0,1]$ by definition. For $\Phi_\delta^{(1)}$, adding $(\lambda, \mu, 0)$ to the argument decreases the numerator by $\lambda$ and the denominator by $2\lambda - 2\mu$. For $\Phi_{1/3}^{(2)}$, the denominator is instead decreased by $\frac32 \lambda - \frac32 \mu$. In both cases, the resulting function value is at most
    \begin{align*}
        \frac{N + \abs{\lambda}}{D - 2\abs{\lambda} - 2\abs{\mu}}
        &= \frac{N}{D} + \frac{\abs{\lambda}}{D - 2\abs{\lambda} - 2\abs{\mu}} + \frac{2N(\abs{\lambda} + \abs{\mu})}{D(D-2\abs{\lambda}-2\abs{\mu})} \\
        &\leq \frac{N}{D} + \frac{\abs{\lambda}}{D - 2\abs{\lambda} - 2\abs{\mu}} + \frac{2(\abs{\lambda} + \abs{\mu})}{D-2\abs{\lambda}-2\abs{\mu}} \\
        &\leq \frac{N}{D} + \frac{6\abs{\lambda} + 4\abs{\mu}}{D},
    \end{align*}
    where the last inequality follows from $2\abs{\lambda} + 2\abs{\mu} \leq \frac{1}{12} \sigmageneric(V_\eta^1) \leq \frac{D}{2}$. Again using the lower bound on $D$, the second term is in $O(\frac{\abs{\lambda} + \abs{\mu}}{\sigmageneric(V_\eta^1)})$. This proves the continuity condition \eqref{eq:continuity_condition}.

    For the monotonicity \eqref{eq:monotonicity_condition}, we observe that by increasing $t$ to some value $t' > t$, we subtract at most $(t'-t) R_0(V_\single)\leq (t'-t)\Opt$ from the numerator and at least $(t'-t)\cdot \sigmageneric(V_\eta^1) $ from the denominator. Since $\Phi_\delta^{(i)}(\bar x(s))$ is always at least $\frac{\Opt}{\sigmageneric(V_\eta^1)}$, this cannot decrease the fraction.
\end{proofof}

\end{document}